\documentclass[journal]{IEEEtran}

\usepackage[final]{graphicx}
\usepackage[reqno]{amsmath}
\usepackage{amssymb}
\usepackage{cite}
\usepackage{balance}
\usepackage{enumerate}
\begin{document}
\title{Throughput Optimal Multi-user Scheduling via Hierarchical Modulation}

\author{Mehmet~Karaca and
        Ozgur~Ercetin
        \vspace{-5 mm}
\thanks{Mehmet Karaca and, Ozgur Ercetin are with the Faculty of Engineering
and Natural Sciences, Sabanci University, 34956 Orhanli-Tuzla,
Istanbul/Turkey. Email: mehmetkrc@su.sabanciuniv.edu,
oercetin@sabanciuniv.edu.}\thanks{This work is supported in part by
Tubitak under grant number 109E242.} }
\newcommand{\argmax}{\operatornamewithlimits{argmax}}

\newtheorem{theorem}{Theorem}
\newtheorem{lemma}[theorem]{Lemma}
\newtheorem{definition}{Definition}
\IEEEoverridecommandlockouts
\maketitle
\begin{abstract}
We investigate the network stability problem when two users are
scheduled simultaneously. The key idea is to simultaneously transmit
to more than one users experiencing different channel conditions  by
employing hierarchical modulation. For two-user scheduling problem,
we develop a throughput-optimal algorithm which can stabilize the
network whenever this is possible. In addition, we analytically
prove that the proposed algorithm achieves larger achievable rate
region compared to the conventional Max-Weight algorithm which
employs uniform modulation and transmits a single user. We
demonstrate the efficacy of the algorithm on a realistic simulation
environment using the parameters of High Data Rate protocol in a
Code Division Multiple Access system. Simulation results show that
with the proposed algorithm, the network can carry higher user
traffic with lower delays.
\end{abstract}
\section{Introduction}
\IEEEPARstart{T}{he} scheduling is an essential problem for any
shared resource. The problem becomes more challenging in a dynamic
setting such as wireless networks where the channel capacity is time
varying due to multiple superimposed random effects such as mobility
and multipath fading. In a queueing system, the most important
property of a scheduling algorithm is to keep the network stable
(e.g., the queue sizes remain bounded over time). The seminal work
by Tassiulas and Ephremides has shown that \textit{Max-Weight }
algorithm  scheduling the user with the highest queue backlog and
transmission rate product at every time slot can stabilize the
network whenever this is possible~\cite{MW}.

It is well-known that Max-Weight algorithm is
\textit{throughput-optimal}, i.e., it stabilizes the network for all
arrival rate vectors that are \textit{strictly} within the
achievable rate region. The performance of Max-Weight algorithm has
been investigated in depth for the single user scheduling case.
However, determining a throughput-optimal algorithm when more than
one users are scheduled simultaneously has not received much
attention.  
In this paper, we propose a modulation-assisted throughput optimal
scheduling algorithm scheme where we employ hierarchical modulation
(HM). 

The authors in~\cite{Cover:broadcast72} are the first to show the
advantage of HM in broadcast systems. In~\cite{Alouini:HM10}, the
authors proposed a multi-user scheduling algorithm and showed that
HM offers lower queueing delay at the transmission buffer. However,
in neither of these works, the authors considered the stability of
the network. Network utility maximization problem with HM was
investigated in~\cite{Karaca}. Meanwhile,~\cite{Neely:energy06}
proposed scheduling and flow control algorithms but did not take
into account the effect of modulation on the performance of the
algorithm.

Our contribution can summarized  as follows: $i)$ we propose a
throughput optimal algorithm, namely Max-Weight with Hierarchical
Modulation (MWHM) when two users are scheduled simultaneously;$ii)$
we give the conditions under HM that should be employed by
considering both analytical and implementation issues. $iii)$ we
prove that the proposed algorithm achieves larger rate region
compared to the conventional Max-Weight algorithm; $iv)$ we develop
a lower complexity version of MWHM algorithm; $v)$ we demonstrate
via realistic simulations that our algorithm not only keeps the
network stable with higher arrival rate but also reduces the average
delay.

\vspace{-5 mm}
\section{System Model}
We consider a cellular system with a single base station (BS)
transmitting to $N$ users. Let $\mathcal{N}$ denote the set of users
in the cell. Time is slotted, $t\in \{0,1,2,\ldots\}$.  Let $T_s$
denote the length of the time slot in seconds. Let $h_n(t)$
represent channel gain of user $n$ at time $t$, $n\in
\{1,2,\dots,N\}$. The gain of the channel is constant over the
duration of a time slot but varies between slots.

HM is one of the techniques for multiplexing and modulating multiple
data streams into one single symbol stream, where those multiple
symbols are superimposed together before transmission. In this
paper, for the sake of ease of exposition we assume that only two
layers of hierarchical modulation is used to serve two users
simultaneously. Specifically, we assume that QPSK/16-QAM HM is
implemented. Let  BS transmit to  two users, i.e., user $n$ and user
$m$, at time slot $t$ by employing two layers of HM. Assume that
user $j$ has a better channel than user $i$ , i.e., $h_n\leq h_m$.
Then, user $n$ is assigned to QPSK constellation which we refer to
as \textit{base layer}. User $m$ is assigned to 16-QAM constellation
which we refer to as \textit{incremental layer}. More information
about hierarchical modulation can be found in~\cite{Alouini:HM10}
and references there in. Since two modulated signals are mixed
before being transmitted, they interfere with each other at the
receiver side. However, in~\cite{Liu08}, the authors propose a
decoding technique to cancel the interference seen at the
incremental layer. Specifically, when mixed signal reaches to the
receivers, the data at the base layer is first decoded and removed.
Hence, the data at the incremental layer does not suffer from  the
transmission at the
base layer. 
In~\cite{Liu08}, the achievable rates for user $n$ and user $m$ are
given respectively as follows:
\begin{align}
\mu_{n}^{b}(t)&=T_s\times BW \times\log\left(1 + \frac{h_n(t)P_{n,b}(t)}{h_n(t)P_{m,i}(t) +\sigma}\right),\label{eq:rate1}\\
\mu_{m}^{i}(t)&=T_s\times
BW\times\log\left(1+\frac{h_m(t)P_{m,i}(t)}{\sigma}\right),\label{eq:rate2}
\end{align}
where $P_{n,b}(t)$ and $P_{m,i}(t)$ are the transmission powers for
user $n$ and $m$ at the base and incremental layers, respectively.
$\sigma$ is the noise power and $BW$ is the bandwidth of the
channel. We assume that the BS transmits at fixed power and the
total power consumption is equal to $P$, i.e., $P_{n,b}(t) +
P_{m,i}(t)= P \quad \forall t$. $\mu_{n}^{k}(t)$ is upper-bounded
such that $\mu_{n}^{k}(t) \leq \mu_{max} \ \forall t$, $k \in
\{b,i\}$. Note that when UM is applied at the physical layer, full
power is assigned to a single user and the amount of data that can
be transmitted in that case is given by,
\begin{align}
\mu_n^{um}(t)&=T_s\times BW \times\log\left(1 +
\frac{h_n(t)P}{\sigma}\right).\label{eq:rate_UM}
\end{align}
where $n \in \{1,2,\ldots,N\}$. Let $a_n(t)$ be the amount of data
(bits or packets) arriving into the queue of user $n$ at time slot
$t$ and $a_n(t)\leq a_{max} \ \forall t$, and assume that $a_n(t)$
is a time and user independent stationary process. We denote the
arrival rate vector as $\boldsymbol
\lambda=(\lambda_1,\lambda_2,\cdots,\lambda_N)$, where $\lambda_n =
{\mathbb E}[a_n(t)]$. Let $
\textbf{q}(t)=(q_1(t),q_2(t),\cdots,q_N(t))$ denote the vector of
queue sizes, where $q_n(t)$ is  the queue length of user $n$ at time
slot $t$. The dynamics of the queue of user $n$ is given as,
\begin{align}
q_n(t+1)=[q_n(t)+a_n(t)-\mu_{n}^{k}(t)]^+ .
\end{align}
where $(x)^+=\max(x,0)$ and $k \in \{b,i\}$. Let $\Lambda$ denote
the \textit{achievable rate region} (or rate region) defined as the
closure of the set of all arrival rate vectors for which there
exists an appropriate scheduling policy stabilizing the network.
\section{Throughput Optimal scheduling }
In this section, we give a throughput optimal scheduling algorithm
when HM is employed. However, we start with the conventional
Max-Weight algorithm  employing UM to schedule a single user at
every time slot.
\\
\textbf{ Max-Weight with UM (MWUM)}: At time $t$, given $h_n(t)$ and
$q_n(t)$ for all $n \in \mathcal{N}$,  schedule user $n^*$ which has
the maximum queue length and service rate product, i.e.,~\cite{MW}:
\begin{align}
w_n^{um}(t)&= q_n(t)\mu_n^{um}(t)\label{eq:weight_um}\\
n^*&=\argmax_{n \in \mathcal{N}} \
w^{um}_n(t)\label{eq:scheduling_um}.
\end{align}
We define $W_u(t)\triangleq w^{um}_{n^*}(t)$. Let $\Lambda_{u}$
denote the rate region achieved by MWUM.

\textbf{Max-Weight with HM (MWHM)}: At time $t$, given $h_n(t)$ and
$q_n(t)$ for all $n \in \mathcal{N}$ schedule \emph{two} users
$(n^*,m^*)$ such that $h_{n^*}(t) \leq h_{m^*}(t)$ to maximize the
sum of queue length and service rate products, i.e.,:
\begin{align}
w_{n,m}^{hm}(t)&= q_n(t)\mu_{n}^{b}(t) +
q_m(t)\mu_{m}^{i}(t)\label{eq:weight_hm}\\
(n^*,m^*)&=\argmax_{\substack{(n,m) \in \mathcal{N},n \neq m\\
h_n(t) \leq h_m(t)}} \ w^{hm}_{n,m}(t)\label{eq:scheduling_hm}.
\end{align}
We define $W_h(t)\triangleq w^{hm}_{n^*,m^*}(t)$. Let $\Lambda_{h}$
denote the rate region achieved by MWHM.

Since BS has limited power budget, power allocation must be
performed to determine $\mu_{n}^{b}(t)$ and $\mu_{m}^{i}(t)$.

\textit{Power Allocation with HM}:
 Recall that Max-Weight type
scheduling algorithms aim to maximize the weight $w_{n,m}^{hm}(t)$
(or $w_n^{um}(t)$) at each time slot. It is easy to determine the
maximum weight achieved by UM, $w_{n^*}^{um}(t)$. However, the
maximum weight under HM depends on power allocations $P_{n,b}^*(t)$
and $P_{m,i}^*(t)$. Without loss of generality, for a given pair of
users, e.g., user $n$ and user $m$ such that $h_n(t) \leq h_m(t)$,
the optimal power allocation maximizing the weight $w_{n,m}^{hm}(t)$
is obtained by solving the following optimization problem:
\begin{align}
\max_{P_{n,b}(t),P_{m,i}(t)}\
&q_n(t)\mu_{n}^{b}(t) + q_m(t)\mu_{m}^{i}(t)\label{eq:objective_hm}\\
&\text{s.t.} \quad  P_{n,b}(t) + P_{m,i}(t) = P
\label{eq:constraint}
\end{align}
Note  that $P_{n,b}^*(t)$ and $P_{m,i}^*(t)$ both have a non-zero
value when \eqref{eq:objective_hm}-\eqref{eq:constraint} is a convex
problem. Now, we give a Lemma which states the necessary conditions
for this to hold. For notational convenience, we drop the time
index. Let us define,
\begin{align*}
A \triangleq \frac{ h_n^2}{( h_nP_{m,i}+ \sigma)^2} \ \text{and} \ B
\triangleq \frac{h_m^2}{( h_mP_{m,i}+ \sigma)^2},
\end{align*}
where $0 \leq P_{m,i}(t) \leq P$.
\begin{lemma}
\label{lemma:2} The problem
\eqref{eq:objective_hm}-\eqref{eq:constraint} is a convex
optimization problem when the following inequality is satisfied
\begin{align}
q_nA \leq q_mB
\end{align}
\end{lemma}
\begin{proof}
We show that the objective function in (\ref{eq:objective_hm}) is
concave under the given condition. The objective function can be
rewritten by noting that $P_{n,b}=P-P_{m,i}$. Since the parameters
$T_s$ and $BW$ do not effect the concavity,  we have the following
objective function,
\begin{align*}
f= q_n\log\left(1 + \frac{h_n(P-P_{m,i})}{h_nP_{m,i} +\sigma}\right)
+ q_m\log\left(1+\frac{h_mP_{m,i}}{\sigma}\right).
\end{align*}
Taking the second derivative of $f$ with respect to $P_{m,i}$
yields,
\begin{align}
\frac{d^2 f}{d P_{m,i}^2}=  q_nA -q_mB.
\end{align}
For concavity, $\frac{d^2 f}{d P_{m,i}^2}$ must be less than or
equal zero, i.e., $\frac{d^2 f}{d P_{m,i}^2} \leq 0$. Thus, $q_nA
\leq q_mB$.
\end{proof}
As long as the condition in Lemma 2 is satisfied, the optimal power
allocation can be found by taking the first derivative of $f$ and
setting it to zero. The first derivative of $f$ with respect to
$P_{m,i}$ is given by,
\begin{align}
\frac{d f}{d P_{m,i}}=-q_n\sqrt A +q_m\sqrt B=0.
\end{align}
Thus, we have,
\begin{align}
P_{m,i}^*&=\frac{\sigma(q_mh_m-q_nh_n)}{h_nh_m(q_n-q_m)},\label{eq:optalloc1}\\
P_{n,b}^*&=P-P_{m,i}^*\label{eq:optalloc2}.
\end{align}

\begin{lemma}
\label{lemma:1} If $\boldsymbol \lambda \in \Lambda_{h}$ (i.e.,
$\boldsymbol \lambda$ is feasible), then MWHM algorithm stabilizes
the network and it is throughput optimal.
\end{lemma}
\begin{proof}
The proof is given in Appendix \ref{sec:lemma1}.
\end{proof}
\section{Max-Weight Algorithm with Dynamic Modulation (MWDM)} Note
that MWHM can only be used when there is an inner point solution to
the problem \eqref{eq:objective_hm}-\eqref{eq:constraint}, i.e., $0
<P_{m,i}^*<P $. If the solution is on the boundary, i.e.,
$P_{m,i}^*=0$ or $P_{m,i}^*=P$, then full power is assigned to a
single user and HM is no longer employed, i.e., transmission to a
single user is optimal.

Now, we propose Max-Weight algorithm with dynamic modulation (MWDM)
that dynamically decides which modulation (HM or UM) must be
employed at every time slot. Let $W_d(t)$ and $\Lambda_d$ be the
maximum weight at time $t$ and the rate region achieved  by MWDM,
respectively. MWDM is implemented as follows:
\begin{itemize}
\item  Step 1: The scheduler applies MWUM and finds the maximum
weight $W_u(t)$ by using \eqref{eq:weight_um} and
\eqref{eq:scheduling_um}.
\item Step 2:  The scheduler applies MWHM as follows:
For every pair of users $(n,m)$, find $w_{n,m}^{hm}(t)$:
\begin{itemize}
\item  if $h_n \leq h_m$, then
user $n$ is embedded at the base layer whereas user $m$ is the
incremental layer or vice versa.
\item Check whether the condition in Lemma 2 is satisfied.

\item If not, $w_{n,m}^{hm}(t)=\max \{
w_{n}^{um},w_{m}^{um}\}$ and UM is employed.

\item Otherwise, determine the optimal power allocation
$P_{m,i}^*$ and $P_{n,b}^*$ according to \eqref{eq:optalloc1} and
\eqref{eq:optalloc2}, respectively.

\item Then determine $w_{n,m}^{hm}(t)$ according to
\eqref{eq:weight_hm}.
\end{itemize}
\item  Step 3: After finding $w_{n,m}^{hm}$ for all pairs, determine
$W_h(t)$.
\item Step 4:  If $W_h(t) > W_u(t)$, then $W_d(t)=W_h(t)$ and HM is
employed. Otherwise, $W_d(t)=W_u(t)$ and UM is employed.
\end{itemize}

Let us define the expected weights achieved by MWDM and MWUM as
${\mathbb E}[W_d(t)]$ and ${\mathbb E}[W_u(t)]$, respectively.
\begin{theorem}
\label{thm:1} The achievable rate region of MWUM algorithm is a
subset of the achievable rate region of MWDM, i.e., $\Lambda_{u}
\subseteq \Lambda_{d}$.
\end{theorem}
\begin{proof}
We use the theorem given in~\cite{Eryilmaz:ToN05} to prove the
lemma. According to the theorem, if ${\mathbb E}[W_d(t) \geq
{\mathbb E}[W_u(t)]$, then MWDM can  at least achieve the rate
region of MWUM. The average weight achieved by MWDM is always
greater than or equal to the weight achieved by MWUM since MWDM can
either apply HM or UM according the maximum weight. Thus, the
following inequality holds at every time slot,
\begin{align}
W_d(t) \geq W_u(t)\label{weights}
\end{align}
Taking the expectation of both sides of \eqref{weights} yields that
${\mathbb E}[W_d(t) ] \geq {\mathbb E}[W_u(t)]$.
\end{proof}
Hence, MWDM can be used to increase the total network throughput.
\vspace{-2 mm}
\subsection{ A low complexity algorithm}
Note that the implementation of MWDM algorithm requires the
calculation of the optimal power allocation and the weight of every
pair of users. This requires a computational complexity of $O(N^2)$.
Now, we propose a low complexity algorithm (L-MWDM) that has a
computational complexity of $O(N)$.

L-MWDM algorithm is implemented as follows:
\begin{itemize}
\item Step 1: Determine user $n^*$
according to \eqref{eq:weight_um} and \eqref{eq:scheduling_um} which
is the optimal user selected under UM. Determine $W_u(t)$ by using
\eqref{eq:weight_um}.

\item Step 2: For every user $m\neq n^*$, $m \in \mathcal{N}$ do:
\begin{itemize}
\item if $h_{n^*} \leq h_{m}$, then user
$n^*$ is embedded at the base layer whereas user $m$ is embedded at
the incremental layer or vice versa.
\item Check the condition in Lemma 2 is satisfied.


\item If it is not, $w_{n^*,m}^{hm}(t)=\max \{
w_{n^*}^{um},w_{m}^{um}\}$ and UM is employed.

\item Otherwise,
determine the optimal power allocation $P_{n^*,b}^*$ and $P_{m,i}^*$
according to \eqref{eq:optalloc1} and \eqref{eq:optalloc2},
respectively. Then, determine $w_{n^*,m}^{hm}(t)$ according to
\eqref{eq:weight_hm}.

\end{itemize}
\item Step 3: After finding $w_{n^*,m}^{hm}(t)$ for every user  $m\neq
n^*$ determine $W_h(t)$.

\item Step 4: If $W_h(t) > W_u(t)$, then $W_d(t)=W_h(t)$ and HM is
employed. Otherwise, $W_d(t)=W_u(t)$ and UM is employed.
\end{itemize}
Note that the difference between MWDM and L-MWDM is that MWDM checks
the weights achieved by every pair of users. Hence, its complexity
increases quadratically with the number of users. On the other hand,
L-MWDM calculates the weights assuming that user $n^*$ is always
scheduled. Thus, its complexity is linear with $N$. However, the
maximum weight obtained with MWDM is always greater or equal to that
of L-MWDM.

\vspace{-2 mm}
\section{Simulation Results}
In our simulations, we model a single cell downlink transmission
utilizing high data rate (HDR) \cite{HDR1}.  The base station serves
20 users and keeps a separate queue for each user. Time is slotted
with length $T_s=1.67$ ms as defined in HDR specifications. We set
$BW=1.25$ MHz, $P=10$ Watts and $\sigma=10^{-6}$ Watts. Packets
arrive at each slot according to Bernoulli distribution. The size of
a packet is set to 128 bytes which corresponds to the size of an HDR
packet. The wireless channel between the BS and each user is modeled
as a \textit{correlated} Rayleigh fading according to Jakes' model
with different Doppler frequencies varying randomly between 5 Hz and
15 Hz.

Figure ~\ref{fig:fig1} depicts the maximum arrival rate that can be
supported by MWDM, L-MWDM and MWUM. Clearly, as the overall arrival
rate exceeds 30 packets/slot queue sizes suddenly increase with MWUM
and the network becomes unstable. However, MWDM and L-MWDM improves
over MWUM by supporting the overall arrival rate of up to 32
packets/slot. Therefore, MWDM can achieve a larger rate region than
MWUM as verified analytically in Theorem \ref{thm:1}. Figure
~\ref{fig:fig1} also shows the sum of the queue lengths vs. mean of
overall arrival rate (packets/slot), i.e., the increase is about 1.2
Mbps. As the average arrival rate increases the average queue
backlogs increase as well in all algorithms. Following Littles' Law,
larger queue backlogs yield longer network delays. However, due to
the possibility of serving more than one queue at a time, MWDM and
L-MWDM outperform MWUM in terms of the average delay. This result
indicates that MWDM and L-MWDM are better techniques for delay
sensitive applications. In addition, the delay performance of MWDM
and L-MWDM are very close. Recall that L-MWDM always schedules user
$n^*$ which is the optimal user in MWUM algorithm. Similar to
L-MWDM, MWDM schedules user $n^*$ most of the time.
\begin{figure}[t]
    \centering
    \includegraphics[width=0.8\columnwidth]{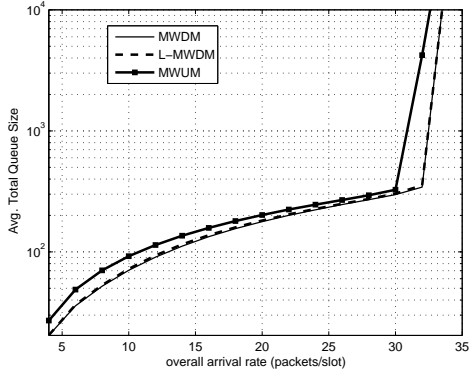}
    \caption{Average total queue sizes vs. overall mean arrival rate.}
    \label{fig:fig1}
\end{figure}
Figure ~\ref{fig:fig2} shows the sum of the queue lengths vs.
transmit power, $P$ in Watts when the overall mean arrival rate is
28 packets/slot. Clearly, as $P$ increases, the average queue size
with both MWDM and MWUM decreases as well. However, the steady state
queue length is achieved with MWUM at $P=4$  whereas MWDM reaches to
the steady state at $P=3$. As a conclusion, HM based scheduling
requires less power than the uniform constellation based scheduling
to stabilize the network.

\vspace{-3 mm}
\section{Conclusions}
In this paper, we investigate the advantages of transmissions of
more than one data streams simultaneously in a network stability
problem. We propose to use hierarchical modulation with Max-Weight
algorithm when two user are scheduled simultaneously. First, we give
the optimal power allocation among users. Then, we show that the
proposed algorithm can support higher user traffic compared to the
conventional Max-Weight Algorithm. In addition, we demonstrate that
with the proposed algorithm the average delay reduces dramatically.
HM is a good technique for scheduling problem especially when the
number of users in the system is large and BS transmits high
transmission powers.

\vspace{-3 mm}

\appendices
\section{Proof of Lemma \ref{lemma:1}}
\label{sec:lemma1} We can write the following inequality by using
the fact $([a]^+)^2 \leq (a)^2, \quad \forall a$:
\begin{align}
q^2_n(t+1) \leq &q^2_n(t) +(\mu_{max})^2+(a_{max})^2\notag\\
&- 2q_n(t)[\mu_n(t)-a_n(t)] \label{eq:queuek}
\end{align}
Define the following Lyapunov function and conditional Lyapunov
drift:
\begin{align}
&L(\textbf{q}(t)) \triangleq \sum_{n=1}^N q_n^2(t),\label{eq:lyapunov-func}\\
&\Delta(t) \triangleq {\mathbb E}\left[ L(\textbf{q}(t+1)) -
L(\textbf{q}(t)) | \textbf{q}(t)\right].\label{eq:lyapunov_drift}
\end{align}
By using \eqref{eq:queuek} and \eqref{eq:lyapunov-func}, one can
show that the Lyapunov drift of the system satisfies the following
inequality at every time slot,
\begin{equation}
\Delta(t) \leq B -
\sum_n\mathbb{E}\left\{q_n(t)\mu_n(t)|\textbf{q}(t)\right\}-\sum_n
q_n(t)\lambda_n \label{eq:Lyp_drift}
\end{equation}
where $B=\frac{N}{2} ((\mu_{max})^2+(a_{max})^2)$.
Note that the second term in the right hand side of
\eqref{eq:Lyp_drift} can be rewritten as follows when two users are
scheduled:
\begin{align*}
\sum_n\mathbb{E}\left\{q_n(t)\mu_n(t)|\textbf{q}(t)\right\}
=\sum_{\substack{(i,j)\\  h_i \leq h_j}}
\mathbb{E}\left\{q_i(t)\mu_i(t) +
q_j(t)\mu_j(t)|\textbf{q}(t)\right\}
\end{align*}
Now, it is easy to see that MWHM minimizes the right hand side of
\eqref{eq:Lyp_drift} at every time slot. Therefore, according to
Lyapunov-Foster criteria, $\textbf{q}(t)$ process is positive
recurrent Markov chain and MWHM can stabilize the network whenever
this is possible~\cite{Georgiadis:Resource06}.
\begin{figure}[t!]
    \centering
    \includegraphics[width=0.8\columnwidth]{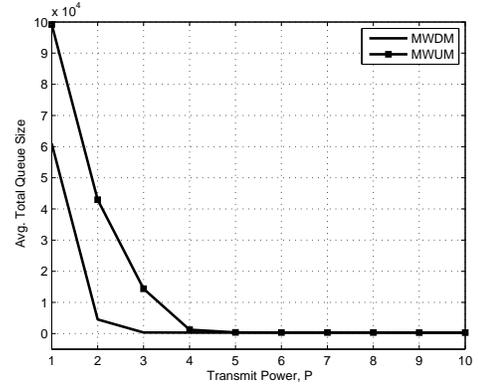}
    \caption{Average total queue sizes vs. transmit power, P}
    \label{fig:fig2}
\end{figure}


\vspace{-3 mm}
\bibliographystyle{IEEEtran}
\bibliography{IEEEabrv,ref}

\end{document}